\documentclass[10pt,twocolumn,twoside]{IEEEtran}
\IEEEoverridecommandlockouts

\usepackage{cite}
\usepackage{amsthm,amsmath,amssymb,amsfonts}
\usepackage{algorithmic}
\usepackage{graphicx}
\usepackage{textcomp}
\usepackage{balance}
\usepackage{xcolor}
\usepackage{booktabs}

\usepackage[utf8]{inputenc}

\usepackage{mathtools}
\newtheorem{theorem}{Theorem}

\usepackage{tikz}
\usetikzlibrary{arrows.meta,bbox}
\usetikzlibrary{shapes}
\usetikzlibrary{patterns}
\usetikzlibrary{decorations.pathreplacing}
\usetikzlibrary{tikzmark}

\newcommand\norm[1]{\lVert#1\rVert}
\newcommand{\minimize}[1]{\underset{{#1}}{\text{minimize}}}
\newcommand{\maximize}[1]{\underset{{#1}}{\text{maximize}}}
\newcommand{\st}{\text{subject to}}
\newcommand{\mb}[1]{\mathbf{#1}}
\newcommand{\mbg}[1]{\boldsymbol{#1}}

\definecolor{red}{RGB}{186, 40, 58}
\definecolor{green}{RGB}{34,139,34}
\definecolor{maincolor}{HTML}{00274C}
\definecolor{blue}{HTML}{004D80}
\definecolor{purple}{RGB}{128,0,128}

\def\BibTeX{{\rm B\kern-.05em{\sc i\kern-.025em b}\kern-.08em
    T\kern-.1667em\lower.7ex\hbox{E}\kern-.125emX}}
\usepackage{lipsum}

\title{\LARGE \bf 
Online Electricity Pricing from Frequency Measurements
}

\author{Xinwei Liu$^\dagger$, Vladimir Dvorkin$^\dagger$
\vspace{-0.5cm}
 \thanks{ 
$^\dagger$The authors are with the Department of Electrical Engineering and Computer Science, University of Michigan, Ann Arbor, MI, USA. Email: {\tt \{xinwei,dvorkin\}@umich.edu}}
}
\begin{document}
\begingroup
\allowdisplaybreaks

\maketitle

\begin{abstract}

Frequency dynamics in power systems reflect active power imbalance in real time, thereby providing an instantaneous signal to inform electricity pricing. However, existing real-time markets operate on much slower timescales and fail to exploit this signal. In this letter, we develop integrated market--frequency dynamics that enable \textit{online} pricing directly from frequency measurements. Representing the real-time market as a dynamic price-discovery process, and integrating this process with the grid frequency dynamics, we derive an explicit price formation mechanism from frequency measurements. This mechanism manifests as a distributed PID-like controller for each generator, where frequency response is driven and remunerated by electricity prices derived solely from local frequency measurements.
\end{abstract}



\section{Introduction}



Electricity as a commodity is inextricably linked to the physical system, where power imbalances directly translate into measurable changes in system frequency. However, electricity pricing operates on substantially slower timescales than frequency dynamics. Motivated by reliability, electricity is traded \textit{offline}, with transactions settled well before operation, from days ahead (bilateral trading, day-ahead markets) to near real-time markets (5--15 minutes before operation) \cite{kirschen2018fundamentals}, and any frequency deviation is addressed in real-time by fast automatic control rather than by markets. This timescale separation decouples price signals from the actual grid state and prevents the delivery of theoretical market efficiency to real-time operations. Indeed, real-time prices fixed for 5--15 minutes fail to internalize the full variability of renewable generation and stochastic loads evolving on subminute timescales.

In this letter, we bridge the gap between real-time markets and operations by developing \textit{online} real-time electricity pricing that acts on frequency timescales and provably delivers the theoretical properties of offline markets in real-time. Our work is inspired by the literature on optimal frequency control \cite{zhang2015achieving,zhao2016unified,mallada2017optimal}, which integrates economic dispatch objectives into automatic generator control, and by recent research on incentive alignment between private (e.g., profit) and system (e.g., stability) objectives in such control \cite{you2025stability}.  We build on these foundations by making the economic meaning of frequency feedback explicit: we derive a frequency-based price signal with an explicit PID structure that reproduces the dual price dynamics of economic dispatch. This price is obtained from local frequency measurements, and frequency/balance restoration is achieved through a distributed generator response to these prices. Furthermore, we optimize the generator response to guarantee the profitability under online electricity pricing.

\textit{Notation:} Bold lowercase $\mb{x}$ (uppercase $\mb{X}$) letters denote vectors (matrices), calligraphic letters $\mathcal{X}$ are reserved for sets. Notation $\mb{x}^\star$ denotes an optimal value, $\mb{x}_{i}$ is the $i^\text{th}$ element of a vector $\mb{x}$, $\Pi_{\mathcal{P}}[\mb{x}]$ denotes the projection of $\mb{x}$ onto the set $\mathcal{P}$, and $\dot{\mb{x}}$ denotes the time derivative of $\mb{x}$.



\section{Integrated Market-Frequency Dynamics}

Consider a vector of fixed generator dispatch $\mb{p}^{\star}$ from the day-ahead market to meet the expected electricity demand $d$. In real time, when the demand deviates by $\delta$, the goal of the real-time economic dispatch (ED) problem is to find a cost-optimal generator regulation $\mb{r}$ that offsets this deviation. This can be formulated as the following optimization problem:
\begin{subequations}\label{prob:ED}
\begin{align}
    \minimize{\mb{r}\in\mathcal{P}}\quad& c(\mb{r})\coloneqq \tfrac{1}{2}\norm{\mb{p}^\star + \mb{r}}^2_{\mb{C}} + \mb{c}^\top (\mb{p}^\star + \mb{r}) \label{ED:cost_function}\\
    \st\quad & \mb{1}^\top (\mb{p}^\star + \mb{r}) - d - \delta = 0\quad : \lambda \label{ED:power_balance}
\end{align}
\end{subequations}
where $\mathcal{P}\coloneqq\{\mb{r} \,|\,\underline{\mb{p}} \leqslant \mb{p}^\star + \mb{r} \leqslant \overline{\mb{p}}\}$ defines the feasible set for generator regulation. The objective is to minimize the real-time generation cost \eqref{ED:cost_function}, subject to generator constraints and active power balance condition \eqref{ED:power_balance}. The dual variable $\lambda$ is the marginal electricity price; for zero demand deviation, the optimal price is that from the day-ahead market $\lambda^{\star}=\lambda^{\text{da}}.$

This problem is typically solved every 5 to 15 minutes, but we aim at solving it on substantially faster timescales. Towards this goal, we analyze its partial Lagrangian function: 
\begin{align}
    \mathcal{L}^{\text{ED}}(\mb{r} ,\lambda) &= c(\mb{r}) + \lambda(\delta - \mb{1}^\top \mb{r}),
\end{align}
where we used the fact that $\mb{1}^\top\mb{p}^\star-d=0.$ The optimal solution to \eqref{prob:ED} can be characterized as the equilibrium point of the following primal-dual dynamical system:
\begin{subequations}\label{system:primal-dual-ED}
\begin{align}
    \dot{\mb{r}} &= \Pi_{\mathcal{P}}\left[-\nabla_{\mb{r}}\mathcal{L}^{\text{ED}}(\mb{r} ,\lambda)\right]= \Pi_{\mathcal{P}}\left[\lambda\mb{1} - \nabla_{\mb{r}}c(\mb{r})\right],\label{r_dynamics}\\ 
    \dot{\lambda} &= \nabla_{\lambda}\mathcal{L}^{\text{ED}}(\mb{r} ,\lambda) =\delta - \mb{1}^\top\mb{r}, \label{lambda_dynamics}
\end{align}
\end{subequations}
where generators continuously update their regulation $\mb{r}$ by projecting their best response to the current price (given their marginal cost $\nabla_{\mb{r}}c(\mb{r})$) onto the feasible set $\mathcal{P}$. The regulation price $\lambda$ is continuously updated according to the active power imbalance, and system \eqref{system:primal-dual-ED} resembles the classic Walrasian t\^atonnement (price adjustment) process \cite{uzawa1960walras}. 

The power imbalance can also be described by the frequency deviation $\omega$. For a copper-plate model, the frequency dynamic is abstracted by the linearized swing equation \cite{kundur2007power}: 
\begin{equation} \label{frequency_D}
\dot{\omega} = \tfrac{1}{M}\left(\mb{1}^\top \mb{r} - D\omega - \delta\right),
\end{equation}
where $M$ and $D$ are system inertia and damping, respectively.  

To integrate the two dynamical systems, consider the Lyapunov function associated with the frequency dynamic: 
\begin{align}
    \mathcal{L}^{\text{FD}} (\omega) &= \tfrac{1}{2DM}(D\omega - \mb{1}^\top \mb{r} + \delta)^2
\end{align}
such that, treating $\mb{r}$ and $\delta$ as exogenous signals, the frequency dynamic can be written as the gradient flow $\dot{\omega}=-\nabla_{\omega}\mathcal{L}^{\text{FD}}(\omega)$, which recovers the swing equation. 

Then, we introduce the composite Lagrangian function of the integrated market-frequency dynamics:
\begin{align}
    \mathcal{L} (\mb{r},\omega,\lambda,\gamma) &\coloneqq \mathcal{L}^{\text{ED}}(\mb{r} ,\lambda) + \gamma\mathcal{L}^{\text{FD}} (\omega) \label{eq:com_lagrangian}
\end{align}
which augments the Lagrangian of the ED problem with a frequency constraint penalized by an auxiliary dual variable $\gamma$. The primal-dual dynamics seeking the saddle point of the composite Lagrangian function \eqref{eq:com_lagrangian} takes the form:
\begin{subequations}\label{system:primal-dual-composite-lag}
\begin{align} 
\dot{\mb{r}} &= \Pi_{\mathcal{P}}\left[-\nabla_{\mb{r}}\mathcal{L} (\mb{r},\omega,\lambda,\gamma)\right] \nonumber\\
&= \Pi_{\mathcal{P}}\left[\lambda\mb{1} - \tfrac{\gamma}{DM}(\mb{1}^\top \mb{r} - D\omega - \delta)\mb{1}- \nabla_{\mb{r}}c(\mb{r})\right],\label{primal-dual-composite-lag:regulation}\\
\dot{\omega} &= -\nabla_{\omega}\mathcal{L} (\mb{r},\omega,\lambda,\gamma) = \tfrac{\gamma}{M}(\mb{1}^\top \mb{r} - D\omega - \delta),\label{primal-dual-composite-lag:freqeuncy}\\
\dot{\lambda}& = \nabla_{\lambda}\mathcal{L} (\mb{r},\omega,\lambda,\gamma) = \delta - \mb{1}^\top \mb{r}, \label{primal-dual-composite-lag:lambda}\\
\dot{\gamma} &= \nabla_{\gamma}\mathcal{L} (\mb{r},\omega,\lambda,\gamma)=\tfrac{1}{2DM}(D\omega - \mb{1}^\top \mb{r} + \delta)^2.\label{primal-dual-composite-lag:gamma}
\end{align}
\end{subequations}

Compared to system \eqref{system:primal-dual-ED}, generators respond to the composite price informed by both active power imbalance via equation \eqref{primal-dual-composite-lag:lambda} and frequency deviation via equation \eqref{primal-dual-composite-lag:gamma} on the same timescale. An important result from \cite{zhao2016unified,zhang2015achieving,mallada2017optimal} is that this system converges to the social optimum. We briefly revisit this result. 
\begin{theorem}[Social optimum] \label{th:socially-optimal-restoration}
The dynamical system \eqref{system:primal-dual-composite-lag} converges to the solution of the optimization problem
\begin{subequations}\label{problem:social-optimum}
\begin{align}
    \minimize{\mb{r}\in\mathcal{P}, \omega}\quad& c(\mb{r}) \\
    \st\quad
    & \mb{1}^\top \mb{r} - \delta = 0,  \label{social-optimum:lambda}\\
    & D\omega - \mb{1}^\top \mb{r} + \delta = 0,  \label{social-optimum:gamma}
\end{align}
\end{subequations}
which simultaneously ensures power balance and zero frequency deviation at the minimum generation cost. 
\end{theorem}


The proof follows similar ideas in \cite{zhao2016unified,zhang2015achieving,mallada2017optimal}, so we only provide the intuition. The dynamical system in \eqref{system:primal-dual-composite-lag} can be seen as a mixed-saddle flow dynamics of the following problem:
\begin{subequations}\label{problem:social-optimum-nonconvex}
\begin{align}
    \minimize{\mb{r}\in\mathcal{P}, \omega}\quad& c(\mb{r}) \\
    \st\quad
    & \mb{1}^\top \mb{r} - \delta = 0,  \quad & \colon \lambda\\
    & \tfrac{1}{2DM}(D\omega - \mb{1}^\top \mb{r} + \delta)^2 = 0,  &\colon \gamma \label{aux1}
\end{align}
\end{subequations}
where the last constraint \eqref{aux1} is always satisfied when enforcing \eqref{social-optimum:gamma} instead. That is, the stationary point of \eqref{problem:social-optimum-nonconvex} is the optimal solution of \eqref{problem:social-optimum}; moreover, if the dynamical system \eqref{system:primal-dual-composite-lag} converges to the stationary point of \eqref{problem:social-optimum-nonconvex}, it also converges to the socially optimal solution of \eqref{problem:social-optimum}, as desired.


Unlike power balance, which can be temporarily violated, the swing equation is a physical law and thus requires the hard constraint  $\gamma=1$. This recovers the original swing equation \eqref{frequency_D} in \eqref{primal-dual-composite-lag:freqeuncy} and simplifies the dynamics in \eqref{system:primal-dual-composite-lag}:
\begin{subequations}\label{system:primal-dual-composite-lag-2}
\begin{align} 
\dot{\mb{r}}&= \Pi_{\mathcal{P}}\left[(\lambda - \tfrac{1}{D}\dot{\omega})\mb{1}- \nabla_{\mb{r}}c(\mb{r})\right],\label{primal-dual-composite-lag-2:regulation}\\
\dot{\omega} &= \tfrac{1}{M}(\mb{1}^\top \mb{r} - D\omega - \delta),\label{primal-dual-composite-lag-2:freqeuncy}\\
\dot{\lambda}& = \delta - \mb{1}^\top \mb{r}, \label{primal-dual-composite-lag-2:lambda}
\end{align}
\end{subequations}
where generators in \eqref{primal-dual-composite-lag-2:regulation} optimally respond to power imbalance and frequency deviations according to marginal generation cost. Theorem \ref{th:socially-optimal-restoration} holds for the reduced dynamics as well.

\section{Online Pricing from Frequency Measurements}

In this section, we establish that the classic Walrasian tâtonnement price adjustment process \cite{uzawa1960walras} realizes within generator frequency control as a cost-optimal, local PID-like controller enabling real-time pricing on frequency timescales. 

\begin{figure*}[t]
    \centering
    \resizebox{0.9\textwidth}{!}{%
\begin{tikzpicture}
\node [draw=black, line width = 0.025cm, rounded corners = 2.5,minimum width=200, minimum height=30, align=center] (price_update) at (-3.5,3) {
$\begin{aligned}\pi(\omega(t)) = -M \omega(t)  - D\!\textstyle\int_{0}^{t}\!\omega(\tau)d\tau - \tfrac{1}{D}\dot{\omega}(t)\end{aligned}$};

\node [draw=black, line width = 0.025cm, rounded corners = 2.5,minimum width=215, minimum height=30, right of = price_update, xshift=25.0em, align=center] (regulation_update) {
$\begin{aligned}
     \mathbf{r}(t+1) = \Pi_{\mathcal{P}}\left[\mathbf{r}(t) + \eta \big(\lambda^{\text{da}} + \pi(\omega(t))- \nabla_{\mathbf{r}}c(\mathbf{r}(t)\big)\right]\!\!
\end{aligned}$
};

\node [draw=red, dashed, line width = 0.025cm, rounded corners = 2.5,minimum width=525, yshift=7.5em, xshift=-3.25em, minimum height=50] (controller) at (2.75,0.4) {};
\node[red,draw=none,fill=white,above of = controller,yshift=-0.25em,xshift=-1.2em,minimum width=77] {\large Controller}; 

\node [draw=blue, dashed, line width = 0.025cm, rounded corners = 2.5,minimum width=270, yshift=-2em, xshift=-3.75em, minimum height=80] (pwr_sys) at (2.75,0.4) {};
\node[blue,draw=none,fill=white,above of = pwr_sys,yshift=1.2em,xshift=-0.1em,minimum width=77] {\large Power Grid}; 

\node [draw=black, line width = 0.025cm, rounded corners = 2.5,minimum width=55, minimum height=40, xshift=14em, yshift=-1.25em, text depth=0em, align=center] (invert) {inverters\\ \includegraphics[width=1cm]{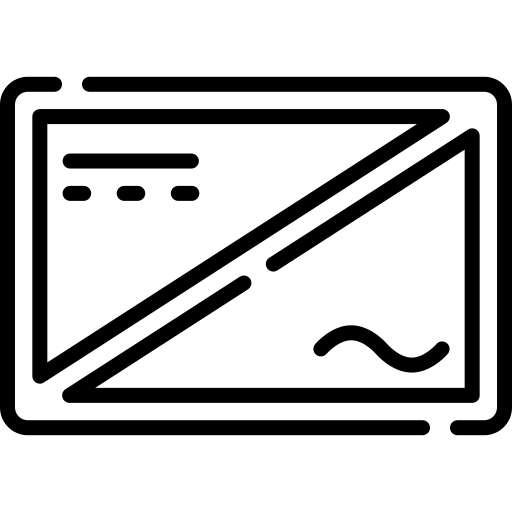}};

\begin{scope}[scale=1,shift={(-2.8,0)}]
        \node[solid,draw=black,thick,circle,fill=black!30,inner sep=1pt,minimum size=10pt] (pn1) at (0,-0.15)  {};
        \node[solid,draw=black,thick,circle,fill=black!30,inner sep=1pt,minimum size=10pt] (pn2) at (1.25,0.25) {};
        \node[solid,draw=black,thick,circle,fill=black!30,inner sep=1pt,minimum size=10pt] (pn4) at (0.75,-0.85) {};
        \node[solid,draw=black,thick,circle,fill=black!30,inner sep=1pt,minimum size=10pt] (pn6_) at (1.85,-1.2) {};
        \node[solid,draw=black,thick,circle,fill=black!30,inner sep=1pt,minimum size=10pt] (pn3) at (3.25,0.45) {};
        \node[solid,draw=black,thick,circle,fill=black!30,inner sep=1pt,minimum size=10pt] (pn6) at (5,0.3) {};
        \node[solid,draw=black,thick,circle,fill=black!30,inner sep=1pt,minimum size=10pt] (pnb) at (6,-0.75) {};
        \node[solid,draw=black,thick,circle,fill=black!30,inner sep=1pt,minimum size=10pt] (pn7) at (4.5,-1.25) {};
        
        \draw[draw=black,thick] (pn1) -- (pn2);
        \draw[draw=black,thick] (pn2) -- (pn3);
        \draw[draw=black,thick] (pn2) -- (pn4);
        \draw[draw=black,thick] (pn1) -- (pn4);
        \draw[draw=black,thick] (pn4) -- (pn6_);
        \draw[draw=black,thick] (pn3) -- (pn7);
        \draw[draw=black,thick] (pn6) -- (pn7);
        \draw[draw=black,thick] (pn6_) -- (pn7);
        \draw[draw=black,thick] (pn6_) -- (pn3);
        \draw[draw=black,thick] (pnb) -- (pn7);
        \draw[draw=black,thick] (pn3) -- (pn6);
        \draw[draw=black,thick] (pn6) -- (pnb);
\end{scope}
\node[draw=black!30,fill=white,opacity=0.8,rounded corners = 2.5,above of = pwr_sys,yshift=-3.25em,xshift=-3em,scale=0.9] {\large $\dot{\omega}(t) = f\big(\omega(t), \mb{r}(t), \delta(t)\big)$}; 

\draw[->,>=stealth,line width = 0.025cm] (price_update) -- node[above] {$\pi(\omega(t))$} (regulation_update);
\draw[black,->,>=stealth,line width = 0.025cm] (invert.east)+(0,1em) -- ++(3em,1em) --++(0em,4.5em) --++(-43em,0em) node[pos=0.825,above] {frequency measurement $\omega(t)$} --++(0em,4.3em) -- (price_update.west) ;
\draw[->,>=stealth,line width = 0.025cm] (regulation_update.east) -- ++(3em,0em) -- ++(0em,-11em) -- node[pos=0.5,above] {regulation setpoint $\mb{r}(t+1)$} ++(-15.6em,0em);
\draw[->,>=stealth,line width = 0.025cm] (-7.25,-0.875) -- node[pos=0.5,above] {disturbance $\delta(t)$} (-2.25,-0.875);
\end{tikzpicture}}
    \caption{Block diagram of the controller. The controller gets the local frequency measurement, computes the frequency-dependent price adjustment to the day-ahead price, then determines and implements the optimal regulation of active power output.}
    \label{fig:main}
\end{figure*}
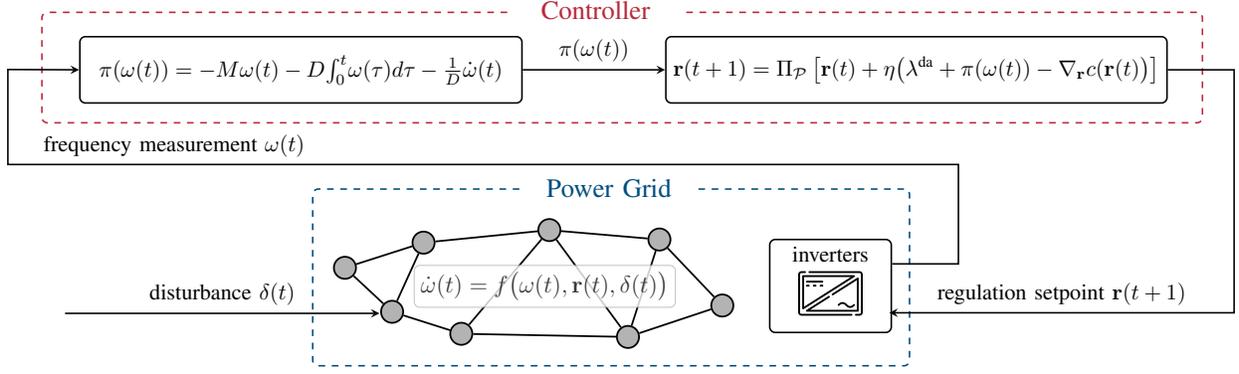

Towards this goal, first observe that the electricity price update in \eqref{primal-dual-composite-lag-2:lambda} is given by the active power imbalance, which can also be expressed from the swing equation \eqref{primal-dual-composite-lag-2:freqeuncy} via frequency variables. The two yield
\begin{align}
    \dot{\lambda} = - D\omega -M\dot{\omega},\label{eq:price_diff}
\end{align}
where the price is driven by both the frequency deviation and the rate of frequency change. Suppose the initial system state corresponds to the day-ahead schedule, so that the price  and frequency deviation are respectively $\lambda(0)=\lambda^{\text{da}}$ and $\omega(0)=0$. Then, integrating both side of \eqref{eq:price_diff} from $0$ to $t$ yields
\begin{align}
    \lambda(t) = \lambda^{\text{da}} - D\textstyle\int_{0}^{t}\omega(\tau)d\tau -M\omega(t).
\end{align}

Substituting this to \eqref{primal-dual-composite-lag-2:regulation} leads to the following best generator response to frequency measurements alone:
\begin{subequations}\label{system:best_response_to_frequency}
\begin{align}
    &\dot{\mb{r}}(t)= \Pi_{\mathcal{P}}\Big[\big(\lambda^{\text{da}} + \pi(\omega(t))\big)\mb{1}- \nabla_{\mb{r}}c(\mb{r}(t))\Big],\label{eq:best_response_to_frequency}\\
    &\;\;\text{with}\;\; \pi(\omega(t)) \coloneqq -\underbrace{M \omega(t)}_{\text{P}}  - \underbrace{D\!\textstyle\int_{0}^{t}\!\omega(\tau)d\tau}_{\text{I}} - \underbrace{\tfrac{1}{D}\dot{\omega}(t)}_{\text{D}}.\label{eq:price_adjustment}
\end{align}
\end{subequations}

This regulation includes the frequency-dependent price adjustment $\pi(\omega(t))$ of the day-ahead price with the following intuition. At nominal frequency (zero frequency deviation), $\pi(\omega(t))=0$ and generators maintain their best response to the day-ahead price $\lambda^{\text{da}}$, and thus the original generation schedule. For negative frequency deviations, $\pi(\omega(t))>0$ adjusts the day-ahead price upwards incentivizing generators to meet the growing demand. For positive frequency deviations, $\pi(\omega(t))<0$ 
reduces the day-ahead price prompting generators to regulate downwards to meet the reduced demand. Regulation \eqref{system:best_response_to_frequency} thus integrates the classic Walrasian t\^atonnement price adjustment process within generator frequency response. 

The price adjustment \eqref{eq:price_adjustment} contains the three canonical components of PID control: a proportional term (P) reacting to instantaneous frequency deviation, an integral term (I) accumulating frequency error, and a derivative term (D) responding to the rate of frequency change. Provided that the day-ahead price and the current system inertia and damping are known, it can be realized locally by generators from frequency measurements alone. The day-ahead prices are publicly announced well ahead of real-time operation.  Recent work on converter probing and measurement-based identification demonstrates that system-level inertia and damping can be periodically estimated and broadcast by system operators \cite{shahraki2026online}.

In practice, regulation is implemented in discrete time:
\begin{subequations}\label{eq:update_discrete_time}
\begin{align}
    \hat{\mathbf{r}}_{i}(t+1) = \mathbf{r}_{i}(t) + \eta_{i} \big(\lambda^{\text{da}} + \pi(\omega(t))- \nabla_{\mathbf{r}}c_{i}(\mathbf{r}_{i}(t)\big),&\\
    \quad\text{followed by}\quad \mathbf{r}_{i}(t+1)= \Pi_{\mathcal{P}_{i}}\Big[\hat{\mathbf{r}}_{i}(t+1)\Big],&
\end{align}
\end{subequations}
for each generator $i$, as illustrated in Fig. \ref{fig:main}. The step size $\eta_{i}>0$ is the design parameter controlling the speed of generator response to frequency-dependent prices. The selection of the step size is economically motivated: small $\eta_{i}$ may prevent generators from responding fast enough to rapidly changing prices, while large $\eta_{i}$ may increase costs faster than electricity price adjusts. Accordingly, the following result establishes that for fully dispatchable generators, the step size can be optimally chosen to guarantee cost recovery at each time step.

\begin{theorem}[Cost recovery] \label{thm2}
The step size $\eta_{i}=\mb{C}_{ii}^{-1}$, inversely proportional to the quadratic cost coefficient, guarantees nonnegative profit for each generator $i$ at every time step.
\end{theorem}
\begin{proof}
We need to establish that
\begin{align}\label{eq:gen_profit}
    \big(\lambda^{\text{da}} + \pi(\omega(t))\big)\big(\mathbf{p}_{i}^\star  + \mathbf{r}_{i}(t)\big) - c_{i}\big(\mathbf{r}_{i}(t)\big) \geqslant 0,
\end{align}
i.e., the profit function  is nonnegative at any time step $t$. For notational convenience, define:
\begin{align}\label{eq:notation}
    &\mathbf{g}_{i}(t) = \mathbf{p}_{i}^\star  + \mathbf{r}_{i}(t),\quad\lambda^{\text{rt}}(t)=\lambda^{\text{da}} + \pi(\omega(t)),
\end{align}
as real-time generation and price, respectively.

We first observe that regulation update \eqref{eq:update_discrete_time} can be obtained as the solution to the following optimization problem:
\begin{subequations}\label{problem:eq_opt}
\begin{align}
    \maximize{\mathbf{g}_{i}}\;&  \Big[ \mathbf{g}_{i}(t) \!+\!\eta_{i} \big[\lambda^{\text{rt}}(t) \!-\! \mathbf{C}_{ii}\mathbf{g}_{i}(t)\!-\!\mathbf{c}_{i}\big]\!\Big]\mathbf{g}_{i}\!-\!\tfrac{1}{2}\mathbf{g}_{i}^{2}\label{eq_opt_obj}\\
    \st\;&\underline{\mathbf{p}}_{i} \leqslant \mathbf{g}_{i} \leqslant \overline{\mathbf{p}}_{i}\quad\colon\underline{\mbg{\alpha}}_{i}, \overline{\mbg{\alpha}}_{i},
\end{align}
\end{subequations}
as verified by its Karush-Kuhn-Tucker conditions. We show that objective \eqref{eq_opt_obj} is nonnegative and use this to bound the generator's profit.

The dual problem associated with problem \eqref{problem:eq_opt} is
\begin{align}
\minimize{\underline{\mbg{\alpha}}_{i}, \overline{\mbg{\alpha}}_{i}\geqslant0}\quad& \tfrac{1}{2} \mathbf{g}_{i}^{\star}(\underline{\mbg{\alpha}}_{i},\overline{\mbg{\alpha}}_{i})^{2} - \underline{\mathbf{p}}_{i}\underline{\mbg{\alpha}}_{i} + \overline{\mathbf{p}}_{i}\overline{\mbg{\alpha}}_{i}\label{problem:dual_eq}
\end{align}
where the quadratic term $\mathbf{g}_{i}^{\star}(\underline{\mbg{\alpha}}_{i},\overline{\mbg{\alpha}}_{i})^2$ comes from the stationarity condition of problem \eqref{problem:eq_opt}, i.e., 
\begin{align*}
    \mathbf{g}_{i}^{\star}(\underline{\mbg{\alpha}}_{i},\overline{\mbg{\alpha}}_{i}) =  \mathbf{g}_{i}(t) \!+\! \eta_{i}\big[\lambda^{\text{rt}}(t) \!-\! \mathbf{C}_{ii}\mathbf{g}_{i}(t) \!-\! \mathbf{c}_{i}\big] \!-\! \underline{\mbg{\alpha}}_{i} \!+\! \overline{\mbg{\alpha}}_{i}.
\end{align*}

The quadratic term in \eqref{problem:dual_eq} is nonnegative. For fully dispatchable generators with $\underline{\mathbf{p}}_{i}=0$, the second term vanishes. The third term is nonnegative by dual feasibility $\overline{\mbg{\alpha}}_{i}\geqslant0$ and $\overline{\mathbf{p}}_{i}$. Thus the dual objective is nonnegative. By strong duality, the primal objective \eqref{eq_opt_obj} is also nonnegative: 
\begin{align}
    \Big[ \mathbf{g}_{i}(t) + \eta_{i} \big[\lambda^{\text{rt}}(t) - \mathbf{C}_{ii}\mathbf{g}_{i}(t) - \mathbf{c}_{i}\big]\Big]\mathbf{g}_{i} - \tfrac{1}{2}\mathbf{g}_{i}^{2} \geqslant 0.
\end{align}

Adding and subtracting the term $\tfrac{1}{2} \mathbf{C}_{ii}\mathbf{g}_{i}^{2}$ corresponding to the quadratic part of the cost function, and rearranging the terms helps isolating the profit on the left-hand side: 
\begin{align}
&\lambda^{\text{rt}}(t)\mathbf{g}_{i} - \tfrac{1}{2} \mathbf{C}_{ii}\mathbf{g}_{i}^2  -  \mathbf{c}_{i}\mathbf{g}_{i} \geqslant \nonumber\\
& - \tfrac{1}{\eta_{i}}\mathbf{g}_{i}(t)\mathbf{g}_{i} + \mathbf{C}_{ii}\mathbf{g}_{i}(t)\mathbf{g}_{i} + \tfrac{1}{2\eta_{i}}\mathbf{g}_{i}^{2} - \tfrac{1}{2}\mathbf{C}_{ii} \mathbf{g}_{i}^{2}.
\end{align}

Setting the step size $\eta_{i}=\mb{C}_{ii}^{-1}$ makes the right-hand side vanish, yielding
\begin{align}
    &\lambda^{\text{rt}}(t)\mathbf{g}_{i} - \tfrac{1}{2} \mathbf{C}_{ii}\mathbf{g}_{i}^2  -  \mathbf{c}_{i}\mathbf{g}_{i} \geqslant 0.
\end{align}
Substituting \eqref{eq:notation} recovers \eqref{eq:gen_profit} as desired.
\end{proof}

\section{Numerical Results}
We present numerical results in a stylized single-area power system with five generators. The nominal demand $d=200$ MW, and the maximum capacity of fully dispatchable generators is set to $50$ MW. The linear cost coefficients are set to $\$27.4/$MWh, and the quadratic costs range from \$0.01 to \$0.015 per MWh$^2$ across the generators. The inertia and damping constants are chosen to be $M=12$ s and $D=35$ MW/Hz. The step size $\eta_i$ is chosen for every generator according to Theorem \ref{thm2}. The frequency dynamics are simulated on a $0.05$ s timescale, and the sampling (measurement) time is $0.25$ s. We test the controller in two settings: constant load disturbance, where step disturbances settle to steady state over time, and time-varying with an evolving stochastic demand process.



\subsection{Constant Disturbance}
\begin{figure}
    \centering
    \includegraphics[width=1\linewidth]{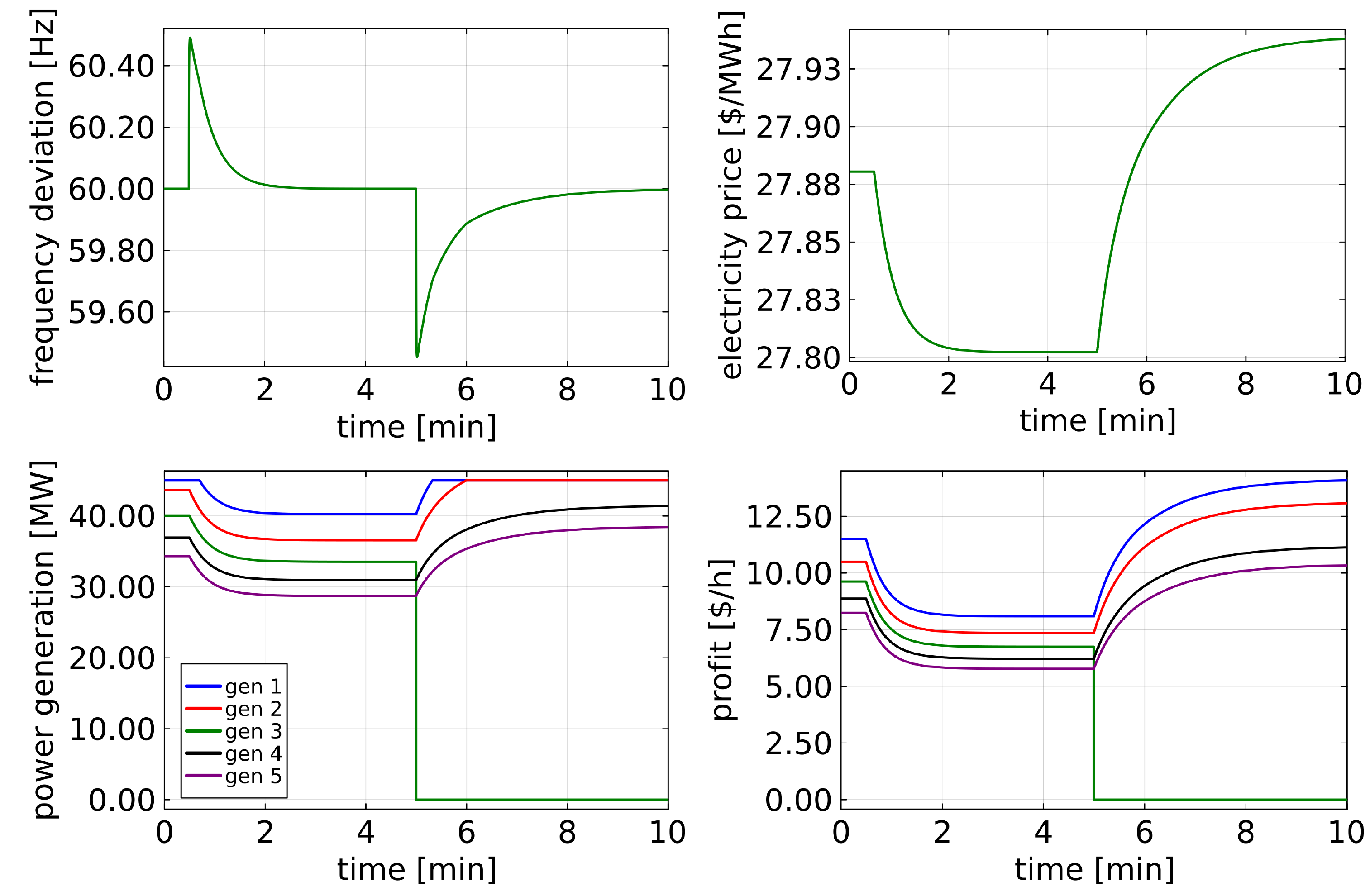}
    \caption{Frequency, electricity price, generation, and profit dynamics.}
    \label{fig:time-invariant}
\end{figure}
We introduce two events: a step decrease in demand of $30$ MW at $t = 30$ s, followed by generator outage at $t = 5$ minutes, both having impacts on frequency and market dynamics as shown in Fig. \ref{fig:time-invariant}. The system frequency is successfully restored after the two events (top-left) as generators respond to dynamic prices (top-right). As expected, the price dynamically decreases after the first event and increases after the second event. The generator response (bottom-left) is consistent with the expected frequency response, and profits (bottom-right) intuitively follow the demand pattern. This correspondence between physical and market variables proves the ability of the integrated dynamics to translate physical measurements into efficient prices, providing sufficient incentives for distributed generators to restore frequency to the nominal value online.

\subsection{Time-Varying Disturbance}
We model the stochastic demand as a Wiener process with zero drift and a standard deviation of 1 MW. Fig. \ref{fig:convergence_transient} illustrates frequency and market dynamics. The controller maintains the frequency around 60 Hz in a time-varying simulation. The top-right plot illustrates the contrast between offline pricing on a 5-minute basis and online pricing of the integrated market-frequency dynamics. Observe that online pricing respects intra-interval frequency changes ignored in the offline market. The bottom-left plot shows resultant generation dynamics: the pale background trajectories are the offline solution, as if we were able to solve economic dispatch at every simulation time step. This illustrates the efficient tracking performance of the controller. The revenues in online and offline settings are given on the bottom-right panel. The online revenues on the frequency (measurement) timescale internalize the variability of frequency dynamics, in contrast to the slow offline revenue dynamics, which are decoupled from actual operation.

\begin{figure}
    \centering
    \includegraphics[width=1\linewidth]{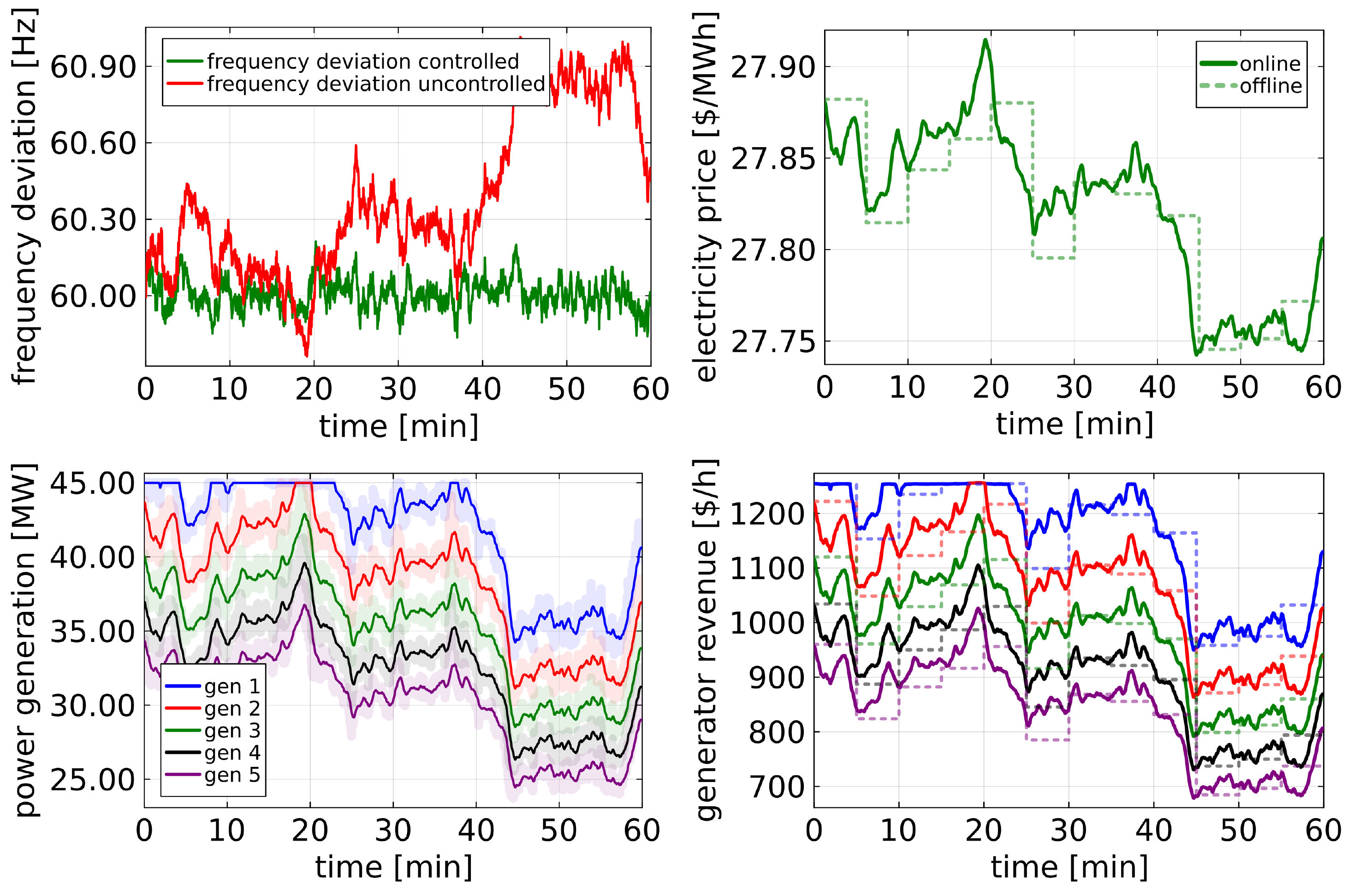}
    \caption{Frequency and market dynamics in time-varying setting.}
    \label{fig:convergence_transient}
\end{figure}

\begin{figure}
  \centering
  \includegraphics[width=1\linewidth]{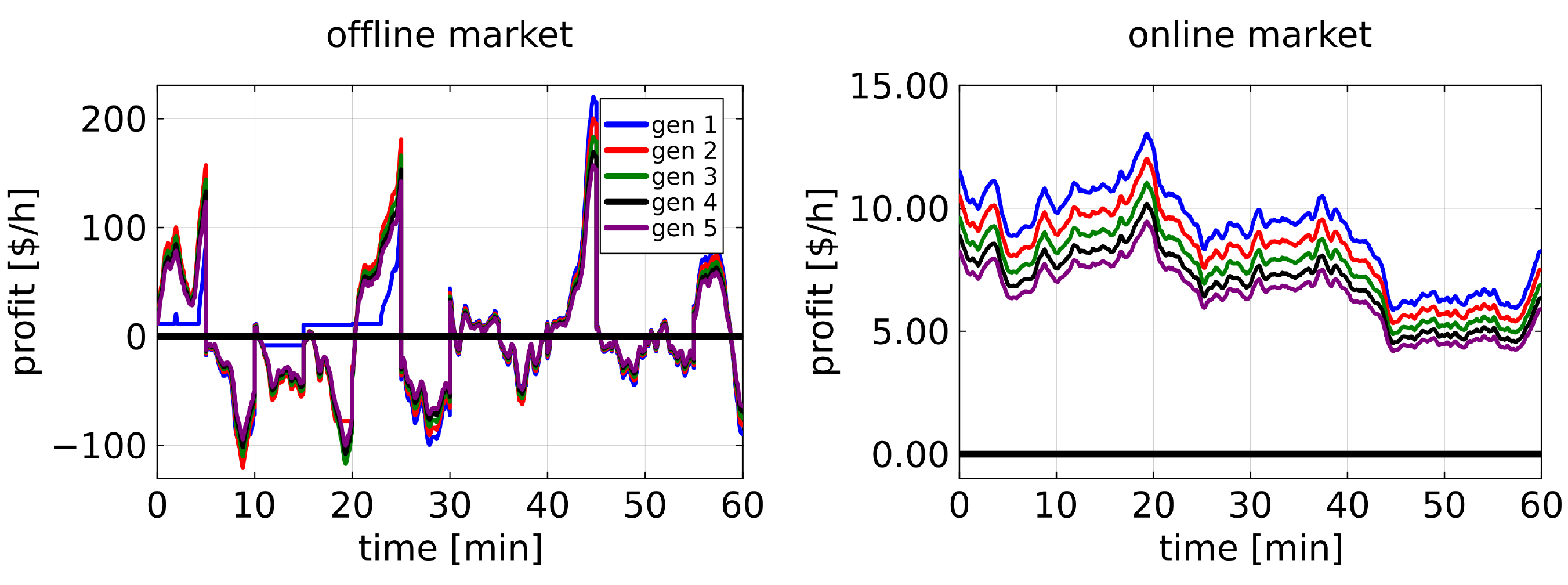}
  \caption{Offline versus online market profits of each generator. }
  \label{fig:profit}
\end{figure}

Lastly, we compare profits in both offline and online market settings in Fig. \ref{fig:profit} for the same frequency response of generators. Under offline prices that do not internalize the per-instance cost of real-time response to frequency deviations, profits are highly volatile, oscillatory, and frequently negative. All are consequences of timescale separation between frequency and market dynamics. Under online pricing, the profit dynamics are more stable and remain in the non-negative domain. These dynamics confirm the result of Theorem \ref{thm2}, ensuring the profitability of response to frequency-derived prices under optimal step-size selection.

\section{Conclusion}
We developed online real-time pricing for electricity on frequency timescales using integrated market-frequency dynamics, realizing it in the form of a local PID controller. Similar to automatic generator control, the sudden changes in load and generation are accommodated in response to frequency deviations. However, in our design, generators respond to local frequency-derived prices that ensure convergence to the social optimum and profitability of frequency response. In future work, we will expand the scope of this letter to transmission network dynamics and focus on computing locational marginal prices from bus frequency measurements. In addition to cost recovery in networked settings, we will focus on establishing the budget balance property of the online market.



\section{AI Usage Disclosure}

The authors used AI tools to assist with writing, including spell-checking and streamlining arguments, as well as for coding and debugging. None of the narrative or code was directly produced by AI; it was used solely in an advisory role. All models and ideas are the sole intellectual property of the authors, and no AI contributed to their development.

\bibliographystyle{ieeetr}
\bibliography{references}
\balance

\endgroup
\end{document}